\def\year{2018}\relax
\documentclass[letterpaper]{article} 
\usepackage{aaai18}  
\usepackage{times}  
\usepackage{helvet}  
\usepackage{amsfonts}
\usepackage[utf8]{inputenc}
\usepackage[english]{babel}
\usepackage{amssymb,amsmath,amsthm}
\usepackage{courier}  
\usepackage{url}  
\usepackage{graphicx}  
\usepackage{color}
\usepackage{multirow}
\usepackage{float}
\usepackage{makecell}
\usepackage{tikz}
\usetikzlibrary{arrows,automata}
\usepackage{balance}
\usepackage{verbatim}
\usepackage[normalem]{ulem}
\usepackage{makecell}
\usepackage{xspace}
\newtheorem{lemma}{Lemma}

\newcommand{\ps}{PS\xspace}

\usepackage{ulem}

\begin{document}

\title{Comparing and Integrating Constraint Programming and Temporal Planning\\ for Quantum Circuit Compilation}
\author{Kyle E. C. Booth$^{\dagger,\ddagger,**,+}$, Minh Do$^{\ddagger,**}$, J. Christopher Beck$^{+}$, \\ \textbf{Eleanor Rieffel$^{\dagger}$, Davide Venturelli$^{\dagger,*}$, \and Jeremy Frank$^{\ddagger}$}\\
$^{\dagger}$Quantum Artificial Intelligence Laboratory, NASA Ames Research Center, Moffett Field, CA\\
$^{\ddagger}$Planning and Scheduling Group, NASA Ames Research Center, Moffett Field, CA\\
$^{*}$USRA Research Institute for Advanced Computer Science, Mountain View, CA\\
$^{**}$Stinger Ghaffarian Technologies, Inc., Greenbelt, MD\\
$^{+}$Department of Mechanical \& Industrial Engineering, University of Toronto, Toronto, ON}

\maketitle

\begin{abstract}
Recently, the makespan-minimization problem of compiling a general class of quantum algorithms into near-term quantum processors has been introduced to the AI community. The research demonstrated that temporal planning is a strong approach for a class of quantum circuit compilation (QCC) problems. In this paper, we explore the use of constraint programming (CP) as an alternative and complementary approach to temporal planning. We extend previous work by introducing two new problem variations that incorporate important characteristics identified by the quantum computing community. We apply temporal planning and CP to the baseline and extended QCC problems as both stand-alone and hybrid approaches. Our hybrid methods use solutions found by temporal planning to warm start CP, leveraging the ability of the former to find satisficing solutions to problems with a high degree of task optionality, an area that CP typically struggles with. The CP model, benefiting from inferred bounds on planning horizon length and task counts provided by the warm start, is then used to find higher quality solutions. Our empirical evaluation indicates that while stand-alone CP is only competitive for the smallest problems, CP in our hybridization with temporal planning out-performs stand-alone temporal planning in the majority of problem classes.
 \end{abstract}

\normalem

\section{Introduction}
\label{sec:introduction}

Quantum computers apply quantum operations, called quantum gates, to qubits, the basic memory unit of quantum processors. Since physical hardware has varying characteristics and architectures, quantum algorithms are often specified as quantum circuits on idealized hardware and must be compiled to specific hardware by adding additional gates that move qubit states to locations where the desired gate can act on them. Compiled circuits that are minimal in duration not only return results more quickly, but
are vital for using near-term quantum hardware that does not support significant quantum error correction or fault tolerance: \textit{quantum decoherence}, an effect that degrades the desired quantum behavior, increases with time. It is, therefore, critical to minimize the duration of the compiled circuit.

Recently, \textit{temporal planning} \cite{fox2003pddl2} was explored to compile quantum circuits \cite{vent:ijcai17}. Gate executions were modeled as durative actions, enabling domain-independent temporal planners to find a valid quantum circuit compilation. Several state-of-the-art planners were used to empirically demonstrate that temporal planning is a promising approach to compile circuits of various sizes to an idealized hardware chip featuring the essential traits of newly emerging quantum hardware.

Historically, \textit{operations research} (OR) techniques have been the primary approach for many combinatorial optimization problems and serve as the backbone of a number of planners. More recently, \textit{constraint programming} (CP) has shown to be competitive with leading OR methods, such as
\textit{mixed-integer linear programming} (MILP), particularly for scheduling problems \cite{Ku16a}. Motivated by this performance, we pose quantum circuit compilation as a scheduling problem, as opposed to a planning problem, where the qubits represent resources and the gates are the tasks to be executed.

In this paper we explore the use of CP for \textit{quantum circuit compilation} (QCC). Our primary contributions are:
\begin{itemize}
\item A stand-alone CP approach for QCC that is competitive with existing temporal planners on small problems, though not for larger problems.
\item A hybrid temporal planning/CP approach where planning solutions are used to warm start the CP solver. Given the same amount of running time, our hybrid outperforms the majority of both stand-alone temporal planning and CP approaches across all solvers and problem classes.
\end{itemize}
In addition to our main contributions, we expand the baseline QCC problem  \cite{vent:ijcai17} to include additional characteristics reflecting realistic hardware architectures. Our expanded benchmarks include: i) less-restricted qubit state initialization,
and ii) crosstalk constraints that place additional restrictions on parallel gate operations.

We consider a refined set of temporal planners for our evaluations. Extensive empirical evaluation shows that the additional constraints lead to a more diverse set of temporal planning benchmarks and the tested temporal planners, which use different planning approaches, perform differently across the problem variations.

The paper is structured as follows: Section~\ref{sec:qcc_background} provides background on QCC and existing solution approaches using domain-independent temporal planners. Section~\ref{sec:cp_planning} describes our CP approach for QCC. Section \ref{section:bounds} discusses details surrounding the planning horizon and the number of tasks. Section~\ref{sec:hybrid} describes our novel temporal planning/CP hybrid approach. Section~\ref{sec:evaluation} details our empirical evaluation and Section~\ref{sec:conclusion} concludes the paper, noting potential future work.

\begin{figure}[tb]
  \includegraphics[width=\columnwidth]{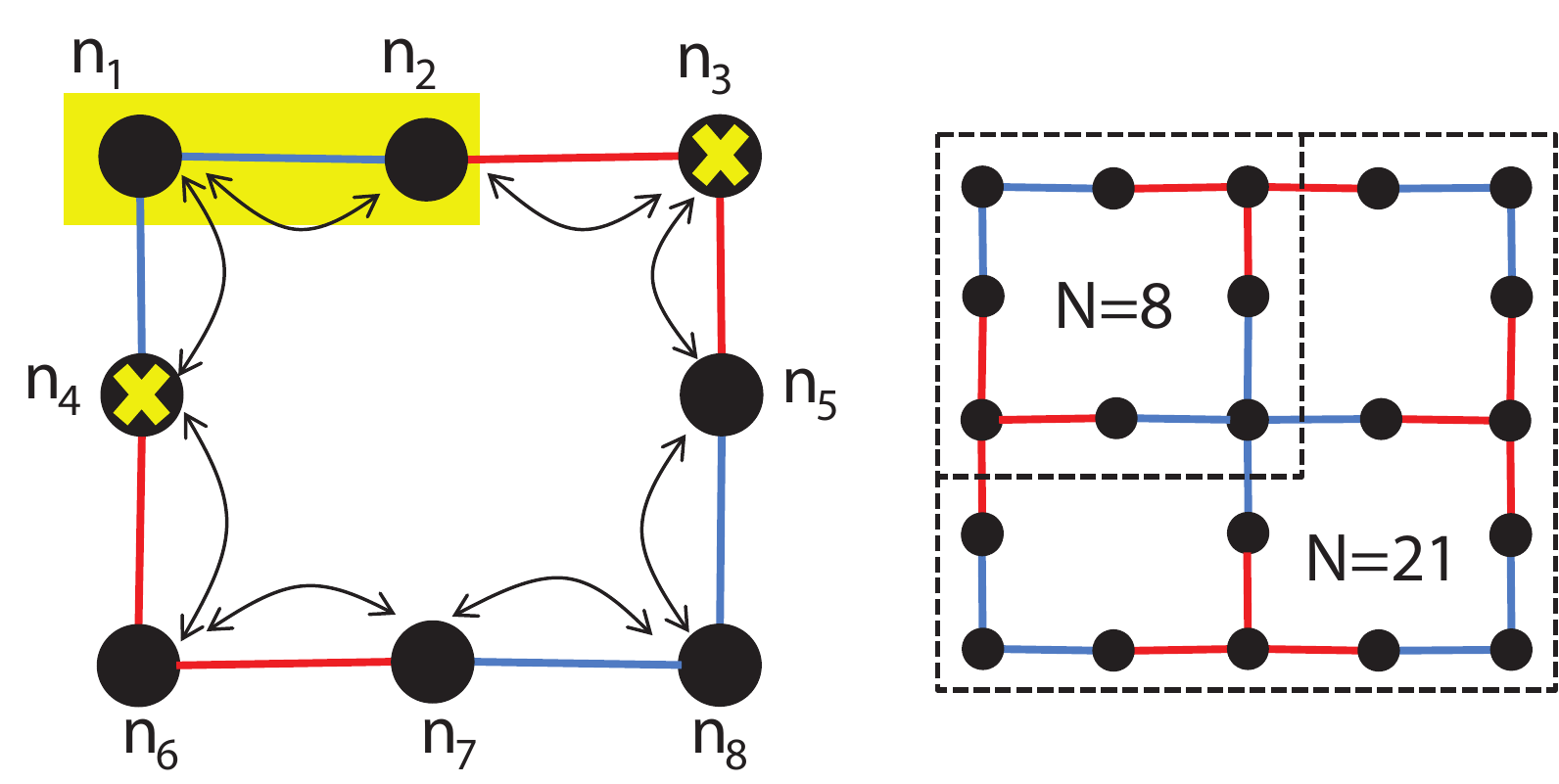}
  \caption{\textit{Left}:
A schematic for the prototype-inspired 8-qubit quantum chip \cite{Sete16} used in previous works \cite{vent:ijcai17}, as well as our numerical experiments. 2-qubit \ps gates are represented by colored edges and swap gates symbolized by edges with double arrows. Each \ps gate color is associated with a distinct duration (three for blue and four for red, in normalized clock cycles). Additionally, 1-qubit mixing gates of unit duration are present at each qubit (graph node). The yellow crosses on $n_3$ and $n_4$ visualize the disabled qubits during the action of a gate between qubits $n_1$ and $n_2$ when crosstalk constraints are present.
\textit{Right}: Dashed boxes indicate the 2 different chip sizes, $N \in \{8,21\}$, used in our empirical evaluation.
}
\label{fig:hardware}
\end{figure}

\section{Quantum Circuit Compilation}
\label{sec:qcc_background}

General quantum algorithms are often described in an idealized architecture in which a gate acts on any subset of qubits. However, in an actual superconducting qubit architecture, such as the ones manufactured by IBM (research.ibm.com/quantum), Rigetti~\cite{2017arXiv170606570R}, Google~\cite{neill2017blueprint} and UC Berkeley~\cite{o2017design}, physical constraints impose restrictions on the sets of qubits that support gate interactions. While there has been active development of software libraries to synthesize and compile quantum circuits from algorithm specifications \cite{Smith16,Steiger16,DeVitt16,Barends16}, few approaches have been explored for compiling idealized quantum circuits to realistic quantum hardware with a specific focus on swap gate insertions \cite{Beals13,Brierly15,Bremner16}, targeting algorithms that could be run in the near-term~\cite{guerreschi2017gate}.

Qubits in these quantum processors can be thought of as nodes in a planar graph, with 2-qubit quantum gates associated with edges and 1-qubit quantum gates associated with nodes. In Figure~\ref{fig:hardware} we present a model chip that is used in our benchmarks. Following the most common choice for benchmarks in the literature, the model quantum algorithm used is a variant of ``Quantum Alternating Operator Ansatz"~\cite{2017arXiv170903489H}, also known as the ``Quantum Approximate Optimization Algorithm" (QAOA)~\cite{Farhi14}, applied to the $\mathcal{NP}$-Hard Max-Cut problem. As described in \cite{vent:ijcai17} this algorithm is specified by a single type of 2-qubit gate, the \textit{phase separation} (\ps) gate, which needs to be applied to a specific set of problem \textit{goals} depending on the instance. Each goal specifies a pair of \textit{qubit states}, the information content of a qubit, that must have a \ps gate applied to them.

In the model chip, \ps gate colors (red or blue) indicate different durations, in terms of clock cycles. Sequences of swap gates, illustrated in Figure~\ref{fig:hardware} with double arrows, are used to achieve goals by moving qubit states to desired locations. Swap gates may only be available on a subset of edges in the graph and swap duration may depend on the edge. In our benchmarks we assume swap gates are available on each edge with constant duration equal to two clock cycles.

Often, it may be desired to apply \ps gates to the problem goals multiple times, where $\mathcal{P}$ is the number of times required. To do this, we separate each \ps goal application with a \textit{mixing} phase in which a single-qubit mixing gate is applied at each qubit.  All \ps  gates that involve a specific qubit state must be carried out before the mixing on that state can be applied and the second \ps stage initiated.  As in previous work \cite{vent:ijcai17}, we consider $\mathcal{P} \in \{1,2\}$.

\subsection{Definitions}

We let the set of qubits in the quantum circuit be represented as $N := \{n_1, n_2,\dots, n_{\alpha}\}$ and the set of qubit states be represented as $Q := \{q_1, q_2,\dots,q_{\beta}\}$. In the problems studied, $\alpha = \beta$. Each qubit, $n_i \in N$, starts in its corresponding (by index) state, $q_j \in Q$, where $i=j$. An integer value $T$ represents the number of time steps (i.e., clock cycles) in the scheduling horizon. Determining an appropriate $T$ is discussed in Section \ref{bound:horizon}.

We let $S$ represent the set of swap gates in the circuit architecture, $S := \{s_1, s_2, \dots, s_{\gamma}\}$, where each gate, $s_k \in S$, involves a qubit pair, $\langle n_i, n_j \rangle$. Similarly, we let $P$ represent the set of \ps gates in the architecture, $P := \{p_1, p_2, \dots, p_{\delta}\}$, where each gate, $p_{\ell} \in P$, involves a pair of qubits. We define the set of swap and \ps gates that involve qubit $n_i \in N$ as $S(i)$ and $P(i)$, respectively. Swap and \ps gates have distinct durations for their activation ($\tau_{swap}$ and $\tau_{\ell}$, respectively), with \ps gate duration depending on the class of the gate, visualized as different colors in Figure \ref{fig:hardware} (thus, duration $\tau_{\ell} \in \{\tau_{red}, \tau_{blue}\}, \forall p_{\ell} \in P$). When the problem involves multiple \ps stages, mixing gates are available at each node in the architecture with a duration $\tau_{mix}$.

The set of problem goals, defined as $G := \{g_1, g_2, \dots, g_{\epsilon}\}$, encode the specific qubit state pairs to which \ps gates need to be applied where each goal, $g_o \in G$, is a pair of qubit states, $\langle q_i, q_j \rangle$. To achieve the goal, these quantum states must be adjacent in the architecture graph (in the case of the studied architecture, all adjacent qubits have a connecting a \ps gate). The \ps gate used for goal activation is a decision variable. \\

\noindent\textit{Example:} Given the 8-qubit architecture in Figure \ref{fig:hardware} with each qubit $n_i \in N$ initially associated to the qubit state $q_j \in Q$ (with $i = j$), let us assume that the idealized circuit requires the application of a \ps gate to the qubit states $q_3$ and $q_4$. The sequence of gates to achieve the goal are:
\begin{eqnarray}
&&\{\textsc{swap}_{n_4,n_1}, \textsc{swap}_{n_2,n_3}\}\rightarrow\textsc{PS}^{blue}_{n_1,n_2}\nonumber \nonumber
\label{examplesequence}
\end{eqnarray}

The sequence takes $\tau_{swap}$ + $\tau_{blue}$ clock cycles as the two swaps can be executed in parallel.

\subsection{Temporal Planning for QCC} QCC problems can be modeled as temporal planning problems, utilizing the standard Planning Domain Definition Language (PDDL) ~\cite{vent:ijcai17}, as follows:
\begin{itemize}
\item Predicates are used to model the location qubit states and if each goal requirement has been achieved or not.
\item Swap and \ps gates are modeled as temporal actions with: i) conditions indicating whether the involved qubit states are residing on adjacent qubits and if the required \ps gates have not been already executed,\footnote{In a solution, each \ps goal gate only needs to be applied once. For planner efficiency, we prevent a \ps gate from executing multiple times through action conditions.} and ii) effects detailing new qubit state location due to swap actions and those that indicate desired \ps goals have been achieved.
\item The objective function is to minimize total plan duration and thus the makespan of the compiled circuit.
\end{itemize}
While the basic mapping is outlined above, there are additional constraints and actions involved with different variations of the QCC problem (e.g., multiple \ps stages). Refer to \cite{vent:ijcai17} for more details.

\subsection{Extensions} In this paper, we target QCC problems beyond the one addressed in ~\cite{vent:ijcai17}. With the addition of \textit{qubit state initialization} (QCC-I) and \textit{crosstalk} (QCC-X) problem variations, our implemented techniques solve a unified problem that originally required two independent steps, incorporating constraints often present in existing hardware.

\subsubsection{Qubit State Initialization (QCC-I)} In the previously studied QCC problem, qubit states are assigned their initial locations on the chip before problem solving (e.g., qubit state $q_j \in Q$ is  assigned to qubit $n_i \in N$). Here, the problem requires finding both the initial assignment of the qubit state locations and the sequence of gates to achieve the goals. Modeling this initialization step in PDDL is as follows.
\begin{itemize}
\item In the initial state all qubits are ``stateless'' and all qubit state locations are undetermined.
\item Action $a_{i,j}$ initializes the location of qubit state $q_j \in Q$ on qubit $n_i \in N$ if: i) $q_j$ has not been initialized, ii) $n_i$ is still empty, and iii) $n_{i-1}$ has been occupied (i.e., initialized). The last condition ensures that qubits are initialized in an arbitrary sequence and thus reduces the total number of valid sequences of initialize actions.
\item Action $a_{init\_finish}$ finalizes the initialization with the condition that all qubit states are located on the graph. After this action, all remaining actions can start.
\end{itemize}

\subsubsection{Crosstalk (QCC-X)} In QCC, gate operations could be applied in parallel provided a qubit was not involved in multiple gates at the same time. For certain hardware architectures, such as the devices manufactured by Google,~\cite{Boixo16}\footnote{Similar constraints are present in the devices by IBM and by UC Berkeley.} crosstalk constraints further restrict qubit involvement. Specifically, when a given qubit $n_i \in N$ is involved in a gate operation, all qubits adjacent to $n_i$ are prevented from engaging in any gate operation. For example, if a 2-qubit gate (either swap or \ps) operation is carried out on $\langle n_1, n_2 \rangle$ in Figure~\ref{fig:hardware}, then no gate operation involving $n_3$ (adjacent to $n_2$) or $n_4$ (adjacent to $n_1$) can be started until the $\langle n_1, n_2 \rangle$ operation is complete.

To model crosstalk constraints in PDDL we introduce:
\begin{itemize}
\item A new predicate $crosstalk(n_i)$ to indicate if $n_i$ is currently disabled by a gate operating on an adjacent qubit.
\item An action representing a gate operation on a pair of adjacent qubits $\langle n_i, n_j \rangle$ will: i) require $(not\ (crosstalk(n_i)) \wedge (not\ (crosstalk(n_j)))$ as durative (i.e., over-all) action conditions and ii) for every qubit $n_k$ that is connected to either $n_i$ or $n_j$, $crosstalk(n_k)$ is the start effect of the action and $(not\ (crosstalk(n_k)))$ is the action's end effect.
\end{itemize}

\section{Constraint Programming for Quantum Circuit Compilation}
\label{sec:cp_planning}

The OR and CP communities have each investigated solving optimization problems closely related to planning. Techniques developed in these fields are also utilized in existing planners as off-the-shelf solvers \cite{minhdo:aips2000,menkes:jair2005,Piacentini18b}, routines to solve sub-problems in decompositions \cite{benton:icaps2012}, models to calculate heuristic values \cite{pommerening2015heuristics,piacentini2017linear}, or as inference techniques customized for planning \cite{vidal:aij2006}.

Motivated by CP's strong performance when applied to scheduling problems, we model QCC with qubits represented as capacitated resources and gate actions as tasks to be scheduled. As opposed to planning formulations, representing QCC as a scheduling problem requires us to specify \textit{a priori} additional problem elements, including a scheduling horizon and bounds on the number of times a \ps, swap, or mixing gate will be used.

\subsection{Decision Variables}

As is common in CP, our formulation uses continuous, integer, optional/mandatory interval, and sequence decision variables. A key limitation to CP technology, as opposed to temporal planning, is that it can only reason over decision variables present in the model. For QCC problems, the number of times a particular swap or \ps gate will be used in a solution plan is unknown \textit{a priori}. We must, therefore, define valid upper bounds for these quantities and introduce the corresponding number of tasks (instantiated as optional interval variables) to the model.\footnote{For example, if the optimal solution to a problem defined on Figure \ref{fig:hardware} used swap gate $\langle n_1, n_2\rangle$ three times, we must ensure \textit{at least} three of these swap gate tasks, at this location, are supplied to the model. Often the number supplied will be more than this, as the bounds are loose. Unused gate tasks are set as absent by the solver.} We define bounds on swap, $\mathcal{U}_{swap}$, and \ps, $\mathcal{U}_{PS}$, tasks for each gate in the architecture in Section \ref{bound:swap}.

Interval variables are a natural way to model swap, \ps, and mixing gate tasks, as they have duration, need to be scheduled, and can be absent or present in a solution. An optional interval variable, $var$, is a rich variable type whose possible values are defined over a convex interval: $ var := \{\bot\} \cup \{[s,e) | s,e \in \mathbb{Z}, s \leq e\}$, where the variable takes on the value $\bot$ if it is not present in the solution\footnote{Mandatory interval variables must be present in the solution.} and $s$ and $e$ represent the start and end values of the interval if it is present in a solution. The variable Pres($var$) takes on a value of 1 if $var$ is present in the solution, indicating the gate it represents is used. Constraints are only active over present interval variables. If present, \text{Start}($var$), \text{End}($var$), and \text{Length}($var$) return the integer start and end times, as well as the length, of the interval variable $var$.

To capture qubit state changes as a result of gate execution, we use sequence variables that represent a total order over a set of interval variables; absent interval variables are not considered in the ordering. This allows for the modeling of relationships such as $\text{Pre}(var)$, which identifies the interval variable preceding $var$ in a candidate solution.

We model the problem in CP with an event-based formulation, tracking qubit state after each gate task that involves that qubit. We define the set of all tasks potentially involving qubit $n_i \in N$ as $E_i$, including a task for state initialization.

The decision variables in our formulation are:

\begin{itemize}
	\item $C_{max}$ \emph{(continuous)}: Makespan of the schedule and objective function value of the formulation, with possible values in $0 \leq C_{max} \leq T$.
    \item $y_{k,m}$ \emph{(optional interval)}: Swap task $m$ for swap gate $s_k \in S$. If present, has a start time Start($y_{k,m}$) $\in [0,T]$ and duration Length($y_{k,m}$) $=\tau_{swap}$. The set of optional swap tasks available for swap gate $s_k \in S$ is defined as: $\bar{y}_k := \{y_{k,1}, y_{k,2},\dots,y_{k,\mathcal{U}_{swap}}\}$.
    \item $z_{\ell,n}$ \emph{(optional interval)}: \ps task $n$ for \ps gate $p_{\ell} \in P$. If present, has start time Start($z_{\ell,n}$) $\in [0,T]$ and duration Length($z_{\ell,n}$) $= \tau_{\ell}$, where $\tau_{\ell} \in \{\tau_{red}, \tau_{blue}\}$, as per the architecture. The set of optional \ps tasks available for \ps gate $p_{\ell} \in P$ is defined as: $\bar{z}_{\ell} := \{z_{\ell,1}, z_{\ell,2},\dots,z_{\ell,\mathcal{U}_{PS}}\}$.
        \item $x_{i,j}$ \emph{(integer)}: State of qubit $n_i\in N$ after task $j \in E_i$, where $E_i := \{\dot{x}\} \cup \{\bar{y}_k : s_k \in S(i)\} \cup \{\bar{z}_{\ell} : p_{\ell} \in P(i)\} \cup \bar{\omega}_i$, and $\dot{x}$ is a dummy task. Each of these variables takes on a value in the set $Q$ of available qubit states, namely $x_{i,j} \in \{1,2,\dots,|Q|\}, \forall j \in E_i, n_i \in N$.
    \item $Z_o$ \emph{(interval)}: Mandatory goal \ps task for goal $g_o \in G$. Start time Start($Z_o$) $\in [0, T]$, duration Length($Z_o$) $\in \{\tau_{red}, \tau_{blue}\}$ and end time End($Z_o$). The makespan objective is the time of the latest completion time of these variables, namely: $C_{max} := \max_{g_o \in G} \big(\text{End}(Z_o)\big)$.
 \end{itemize}

In problems with two \ps stages, we include the following additional decision variables for mix gates:
 \begin{itemize}
    \item $\omega_{i, j}$ \emph{(optional interval)}: Task for mixing qubit state $q_j \in Q$ at qubit $n_i \in N$. If present, has start time Start($\omega_{i, j}$) $\in [0,T]$ with duration Length($\omega_{i,j}$) $=\tau_{mix}$. The set of optional mixing tasks available for qubit $n_i \in N$ is defined as: $\bar{\omega}_{i} := \{\omega_{i,1}, \omega_{i,2},\dots,\omega_{i,\beta}\}$.
    \item $\Omega_j$ \emph{(interval)}: Mandatory mixing task for qubit state $q_j \in Q$. Has start time Start($\Omega_j$) $\in [0,T]$ duration Length($\Omega_j$) $=\tau_{mix}$, and end time, End($\Omega_j$), representing when the mixing of state $q_j \in Q$ is complete.
\end{itemize}

\subsection{Formulation, Objective, and Constraints}
With the problem parameters, decision variables, and associated domains defined, we detail our event-based CP formulation in Figure \ref{fig:cp_qcc}. Constraints (\ref{cp_obj} - \ref{cp_symmetry1}) are required for one and two-stage \ps problems, while Constraints (\ref{mixAlternative} - \ref{goalPrec2}) are only required for two-stage \ps problems.

Objective (\ref{cp_obj}) represents the problem objective which is to minimize the makespan, $C_{max}$, of the circuit. The secondary objective, reduced in weight by a small value $\xi$, minimizes the number of swap tasks.\footnote{This value is subtracted from the objective when comparing to temporal planning approaches for consistency.} The addition of this component was found to improve solver performance while remaining a reasonable objective for QCC problems. Constraint (\ref{cp_initial}) initializes the qubit states to their required initial values and Constraint (\ref{cp_obj2}) requires that solution $C_{max}$ be greater than the end time of all goal variables.

We use a number of \textit{global constraints} \cite{van2006global} defined over a set of variables and encapsulating frequently recurring combinatorial substructure. Constraint (\ref{cp_noOverlap}) uses the \textit{NoOverlap} global constraint \cite{baptiste2012constraint} to perform incomplete, efficient domain filtering on the start times of the interval variables. The model treats each qubit, $n_i \in N$, as a unary capacity resource and ensures that swap, \ps, and mix gates are activated in such a way that two gates involving the same qubit are never active at the same time.

Constraint (\ref{cp_alternative}) makes use of the \textit{Alternative} global constraint \cite{laborie2009ibm}, which links interval variables to a set of optional interval variables, enforcing that only one variable from the optional set can be present and the start time must coincide with the mandatory variable. We use this constraint to maintain the relationship between the goal variables, $Z_o$, and the optional \ps variables, $z_{\ell,n}$. For each \ps gate, the $z_{\ell,n}$ tasks are ordered such that they coincide with a single goal, and thus $|G| = |\bar{z}_{\ell}|, \forall p_{\ell} \in P$. Each goal activates a single \ps gate task across the set of \ps gates.

Constraint (\ref{cp_swap}) implements qubit state updates when a swap interval variable is present, swapping the states of the qubits involved in the corresponding physical swap gate. The term $pre_i(y_{k,m})$ returns the task previous to swap task $y_{k,m}$ in the sequence for qubit $n_i \in N$, allowing the modeling of qubit state swap between the qubit pair, $\langle i, j \rangle$, involved in gate $s_k \in S$. Constraint (\ref{cp_ps}) models a similar relationship for \ps gate tasks. We note that while swap tasks result in an exchange of states between qubits $n_i$ and $n_j$, after a \ps task qubit states remain unchanged.

\begin{figure}[tb]
\begin{align}
& {\text{Minimize:}} \nonumber \\
& {\displaystyle C_{max} + \xi \cdot \sum\nolimits_{s_k \in S} \sum\nolimits_{y_{k,m} \in \bar{y}_{k}} \text{Pres}(y_{k,m})} \label{cp_obj} \\
& {\text{Subject to:}} \nonumber \\
& {\displaystyle x_{i,0} = i, \quad \forall i \in \{1, 2,\dots,\alpha\}} \label{cp_initial} \\
& {\displaystyle C_{max} \geq \text{End}(Z_{o}), \quad \forall g_o \in G} \label{cp_obj2} \\
& {\displaystyle \text{NoOverlap}(E_i), \quad \forall n_i \in N} \label{cp_noOverlap} \\
& {\displaystyle \text{Alternative}(Z_{o}, [z_{1, o}, \dots, z_{\delta, o}]), \quad \forall g_o \in G} \label{cp_alternative} \\
& {\displaystyle \text{Pres} (y_{k,m}) \rightarrow (x_{i,{y_{k,m}}} = x_{j,{pre_j(y_{k,m})}})} \nonumber \\
& {\displaystyle \quad \wedge \ (x_{j,{y_{k,m}}} = x_{i,pre_i(y_{k,m})}),} \nonumber \\
& {\displaystyle \quad \forall y_{k,m} \in \bar{y}_{k}, \langle i,j \rangle \in s_k, s_k \in S} \label{cp_swap} \\
& {\displaystyle \text{Pres} (z_{\ell,n}) \rightarrow (x_{i,{z_{\ell,n}}} = x_{i,{pre_i(z_{\ell,n})}})} \nonumber \\
& {\displaystyle \quad \wedge \ (x_{j,{z_{\ell,n}}} = x_{j,pre_j(z_{\ell,n})}),} \nonumber \\
& {\displaystyle \quad \forall z_{\ell,n} \in \bar{z}_\ell, \langle i,j \rangle \in p_{\ell}, p_\ell \in P} \label{cp_ps} \\
& {\displaystyle \text{Pres} (z_{\ell,n}) \rightarrow (x_{i,{z_{\ell,n}}} = g_{\ell,1} \wedge x_{j,{z_{\ell,n}}} = g_{\ell,2}) }\nonumber \\
& {\displaystyle \quad \vee \ (x_{i,{z_{\ell,n}}} = g_{\ell,2} \ \wedge \ x_{j,{z_{\ell,n}}} = g_{\ell,1}),}\nonumber \\
& {\displaystyle \quad \forall z_{\ell,n} \in \bar{z}_\ell, \langle i,j \rangle \in p_{\ell}, p_\ell \in P} \label{cp_goal} \\
& {\displaystyle \text{Pres}(y_{k,m}) \geq \text{Pres}(y_{k,m+1}),} \nonumber \\
& {\displaystyle \quad \forall y_{k,m} \in \bar{y}_{k} \setminus y_{k,\mathcal{U}_{swap}}, s_k \in S} \label{cp_symmetry1} \\
& {\displaystyle \text{Alternative}(\Omega_{j}, [\omega_{1,j}, \dots, \omega_{{\alpha}, j}])}, \quad \forall q_j \in Q \label{mixAlternative} \\
& {\displaystyle \text{Start}(\Omega_j) \geq \text{End}(Z_o)}, \quad \forall g_o \in G(j), q_j \in Q \label{goalPrec}\\
& {\displaystyle \text{End}(\Omega_j) \leq \text{Start}(Z_o)}, \quad \forall g_o \in G'(j), q_j \in Q \label{goalPrec2}
\end{align}
\caption{CP Model for QCC.} \label{fig:cp_qcc}
\end{figure}

Constraint (\ref{cp_goal}) ensures that if a particular \ps gate task, $z_{\ell,n}$, is present,  the states of the qubits involved
match the corresponding goal. The term $g_{\ell,1}$ represents the first qubit state required by goal $g_{\ell} \in G$, and $g_{\ell,2}$ the second.

To remove some of the symmetries in the model, Constraint (\ref{cp_symmetry1}) specifies that homogeneous optional swap tasks must be used lexicographically.

For problems that have two stages of phase separation, a mixing gate must be applied to each qubit state after all the goals that utilize that state are achieved, and then the goals must be repeated. To achieve this, we introduce a second goal set, $G' := \{g_{\epsilon+1}, g_{\epsilon+2}, \dots, g_{2\cdot\epsilon}\}$, which duplicates the first. We let the sets $G(j)$ and $G'(j)$ denote the goals from $G$ and $G'$, respectively, that involve qubit state $q_j \in Q$. Constraint (\ref{mixAlternative}) ensures that only one of the optional mixing tasks is used for each mixed qubit state and Constraints (\ref{goalPrec}) and (\ref{goalPrec2}) ensure that the mixing tasks separate the two \ps goal sets. Additionally, goal requirements are amended to include the duplicated goal set, $G'$.

\subsection{Qubit Initializations and Crosstalk}

The QCC problem extensions introduced in Section \ref{sec:qcc_background} require minor alterations to the CP model.

\subsubsection{QCC-I}
In the baseline and crosstalk variants, Constraint (\ref{cp_initial}) is applied unchanged. However, in the qubit initializations problem variant it is replaced with the following:
\begin{equation}
\text{AllDifferent}(x_{1,0}, x_{2,0}, \dots, x_{\alpha, 0}) \label{init_alldiff}
\end{equation}
The removal of Constraint (\ref{cp_initial}) allows the solver to select initial values for qubit states, and the addition of Constraint (\ref{init_alldiff}) enforces that the initial states on all the qubits be different, ensuring that all qubit states are present on the chip.

\subsubsection{QCC-X}

In the crosstalk variant of the problem, the qubits that can simultaneously participate in gate activations are further constrained. Constraint (\ref{cp_noOverlap}) is adjusted such that the sets $S(i)$ and $P(i)$ for a given qubit $n_i \in N$ include the gates that involve adjacent qubits to $n_i$ as well.

\section{Setting Bounds}\label{section:bounds}

As before, bounds on scheduling horizon and the number of times a particular gate is used is unknown beforehand. In this section we detail how these values are determined.

\subsection{Scheduling Horizon} \label{bound:horizon}
Our CP formulation can be implemented using a horizon set to infinity, however, it was observed that smaller horizon values improve performance.

Let $\psi$ be the length of the side of the chip ($\psi = 3$ for 8-qubits, and $\psi = 5$ for 21-qubits), $\tau^{max}_{PS} = max(\tau_{red},\tau_{blue})$ be the maximum \ps gate duration, and $\phi = (2 \cdot \psi) - 3$ be the maximum number of swaps required to bring any two qubit states to a pair of adjacent qubits.
\begin{lemma}
For $\mathcal{P}=1$ problems, $T = |G| \cdot (\phi \cdot \tau_{swap} + \tau^{max}_{PS})$ is an upper bound on the optimal makespan.
\end{lemma}
\begin{proof}
There are two components to achieving any goal: i) moving the required qubit states to adjacent qubits, and ii) applying a \ps gate task. For a single goal, the worst case scenario would have the required states located on the opposite sides of the architecture (e.g., located on $n_1$ and $n_8$ in Figure \ref{fig:hardware}), requiring a minimum of $\phi$ swaps to place them next to each other. Then, we must apply a \ps gate, which, in the worst case, will take a duration of $\tau^{max}_{PS}$. We can perform all tasks for $|G|$ goals in sequence, leading to a makespan no worse than: $T = |G| \cdot (\phi \cdot \tau_{swap} + \tau^{max}_{PS})$.
\end{proof}
We use this $T$ value as the horizon parameter in all of our CP experiments for single \ps stage problems.\footnote{An extension of this proof is used to yield a scheduling horizon bound applicable to $\mathcal{P}=2$.}

\subsection{Swap and PS Gate Tasks} \label{bound:swap}
For the scheduling formulation, we determine the number of activation tasks to be allocated per physical swap gate, $\mathcal{U}_{swap}$. If we consider achieving each goal sequentially, we could potentially have to move a qubit state through the entire architecture to become adjacent to the other qubit state. In this case, each swap gate is used once for each goal, yielding $\mathcal{U}_{swap} =|G|$. Note that, although this value is observed to perform well empirically, the circuit can work towards goals in parallel, leading (potentially) to optimal solutions that use more than $|G|$ swaps per physical gate; we will explore this in future work.

In contrast to swap tasks, our CP formulation allocates exactly one \ps task to each physical \ps gate for each goal, $g_o \in G$, resulting in $\mathcal{U}_{PS} = |G|$. This allows for the case that all of the goals are achieved using the same physical \ps gate. Note that the \ps gate tasks are ordered such that the first task corresponds to the first goal, and so forth (as noted in the CP model).
\section{Hybrid Approach with Temporal Planning and Constraint Programming}
\label{sec:hybrid}

Temporal planners find satisficing plans for most of the instances, often fairly quickly. The CP model, while performing well on small problems, struggles to find solutions for larger problems due to high levels of task optionality, reducing the inference that can be performed. For the problem instances investigated, the best solutions found used only about $10\%$ of the \ps and swap tasks allocated to the model (the remainder being set as absent).

\paragraph{Warm Start}
Leveraging the temporal planner's ability to find quality solutions early on in the search, we integrate the two techniques via a warm start procedure, a common boosting technique in OR where solutions found by one solver (e.g., the temporal planner) form a starting point for another (e.g., CP's branch-and-infer) \cite{kramer2007understanding,beck2011combining}.

While the details of how CP Optimizer makes use of a starting solution are sparse \cite{Laborie2013presentation}, there are two general ways in which it is exploited. First, the existence of an upper bound on the objective function allows the solver to propagate and remove possible values from the domains of the decision variables. For example, the bound on the length of the horizon allows many optional interval variables to be removed and the domains of the start time variables of the remaining ones to be narrowed. Second, the solution can be used heuristically to guide the search toward promising areas of the search space (e.g., \cite{Beck07b}). As a complete search technique, given enough time the CP solver will explore areas far from the warm start solution and is guaranteed to find and prove optimality.

\paragraph{Solution Mapping}
With our CP model and temporal planning solution for a problem instance, expressed as time-stamped PDDL actions, we create a mapping of the planning solution to the corresponding variable values in the CP model. In addition to assigning $C_{max}$ to the makespan found by the planner, we map the time-stamped swap, \ps, and mix gate actions to values for the CP model's interval variables and qubit state variables.
After the mapping of all gate actions, we set the remaining interval variables to absent; this maps values to all $y_{k,m}$, $z_{\ell,n}$, and $Z_{o}$ variables in the CP model. In the final mapping step, we assign values to the qubit state variables, $x_{i,j}$, by reasoning about the present (and absent) gate task interval variables, constructing a complete solution to the CP formulation. A similar mapping is applied to $\mathcal{P}=2$ problems, assigning values to the mixing variables, $w_{i,j}$ and $\Omega_j$.

\paragraph{Hybrid Implementations}

The warm start procedure yields a way to link the temporal planning and CP, however, we must specify when the solution is passed between the solvers. We describe two sets of hybrid experiments below.

\medskip
\noindent
\textit{Last Hybrid}: As a proof-of-concept, our first hybridization represents a best-case scenario and explores whether the highest quality solution found by temporal planning, within the runtime limit, can be improved by our CP model. 
We run the temporal planner until the runtime limit, $\mathcal{T}$, storing the last solution found and the time it was found, $t_{last}$. We then warm-start CP with this solution and run it for the remaining time, $\mathcal{T} - t_{last}$. Since, in practice, we would need an oracle to identify the best (last) solution, and thus $t_{last}$, prior to reaching $\mathcal{T}$, these experiments are used to provide an estimate for the best possible performance of our temporal planning/CP hybridization, given a fixed runtime $\mathcal{T}$.

\medskip
\noindent
\textit{Half Hybrid}: Our second hybridization allocates half of the runtime to temporal planning, and half to CP. The best solution found by temporal planning at $\frac{\mathcal{T}}{2}$ is used to warm-start CP, which is then run for the remaining runtime, also $\frac{\mathcal{T}}{2}$. This na\"ive approach is simple to implement and designed with the assumption that temporal planning and CP are equally valuable to the hybrid, thus allocated the same solve time.

\setlength{\tabcolsep}{.31em}
{\renewcommand{\arraystretch}{0.85}
\begin{table*}[ht]
\centering
\scriptsize
\begin{tabular}{|c|c|l|l||c||c|c|c||c|c|c||c|c|c||c|c|}
\hline
\multicolumn{4}{|c||}{Instance \& Evaluation} & CP & \multicolumn{3}{c||}{LPG} & \multicolumn{3}{c||}{TFD} &  \multicolumn{3}{c||}{POPF} & \multicolumn{2}{c|}{CPT} \\
\hline
${|N|}$ & $\mathcal{P}$ & Variant & Eval. & Alone & Alone & Last  & Half & Alone & Last & Half & Alone & Last & Half &  Alone & Hybrid  \\
 \hline
 \hline
\multirow{2}{*}{8} & \multirow{2}{*}{1} & \multirow{2}{*}{QCC} & Score & 0.88 & 0.93 & 0.93& 0.94 & 0.87& 0.93 & 0.92  & 0.83  & 0.85  & 0.85  & 1.00 (50$^{\dagger}$)  & $-$  \\
 & & & $\Delta$ & $-$  & $-$  & 1.4\% (11)  & 0.7\% (10) & $-$ & 6.8\% (30) & 6.0\% (29) &  $-$  & 9.7\% (41)  & 8.3\% (37) & $-$  & $-$  \\
 \hline
\multirow{2}{*}{8} & \multirow{2}{*}{1} & \multirow{2}{*}{QCC-I} & Score &  0.70 (42) & 0.85 & 0.91 & 0.88 & 0.89 & 0.96 & 0.96 & 0.84 & 0.86 & 0.86 & 1.00 (43$^{\dagger}$) & 1.00  \\
 &  &  & $\Delta$ & $-$ &$-$ & 6.6\% (30) & 3.5\% (25) &$-$ & 7.4\% (31) & 7.0\% (30)& $-$ & 7.8\% (32) & 6.9\% (29) & $-$  & 0.00\% (0)   \\
 \hline
\multirow{2}{*}{8} & \multirow{2}{*}{1} & \multirow{2}{*}{QCC-X}& Score & 0.90 & 0.70 & 0.95 & 0.95 & 0.91 & 0.97 & 0.96 & 0.77 & 0.79 & 0.79 & 0.00 (0)  & $-$  \\
 &  &  & $\Delta$ & $-$ & $-$ & 26.8\% (50) & 26.8\% (50) & $-$&  5.8\% (32) & 5.0\% (30) & $-$ & 18.6\% (49)  &  19.8\% (49)&  $-$ & $-$ \\
 \hline
 \hline
\multirow{2}{*}{8} & \multirow{2}{*}{2} & \multirow{2}{*}{QCC} & Score & 0.73 & 0.78 & 0.86 & 0.85 & 0.92 & 0.98 & 0.97 & 0.85 & 0.87 & 0.86 & 0.00 (0) & $-$ \\
 &  &  & $\Delta$ & $-$& $-$ & 8.7\% (42) & 7.6\% (45)& $-$&  6.3\% (40)  &5.3\% (38)  & $-$ & 9.9\% (45) & 7.6\% (45) & $-$ &  $-$  \\ \hline
\multirow{2}{*}{8} & \multirow{2}{*}{2} & \multirow{2}{*}{QCC-I} & Score & 0.00 (0) & 0.61 & 0.72  & 0.70 & 0.88 & 0.97 & 0.96 & 0.77 & 0.79 & 0.79  & 0.00 (0) &  $-$ \\
 &  &  &$\Delta$ & $-$&  $-$& 14.9\% (47) &12.4\% (44)& $-$& 10.0\% (40)& 7.4\% (36) & $-$ & 12.9\% (48) & 11.9\% (46) &  $-$& $-$  \\ \hline
\multirow{2}{*}{8} & \multirow{2}{*}{2} & \multirow{2}{*}{QCC-X} & Score & 0.68 & 0.53 & 0.82 & 0.81 & 0.87 & 0.99  & 0.98 & 0.64 (43) & 0.66 (43) & 0.67 (42)  & 0.00 (0) &  $-$  \\
 &  &  & $\Delta$& $-$ &  $-$& 34.4\% (47) & 34.2\% (50) & $-$ & 12.6\% (50) &11.5\% (49) &  $-$ & 24.2\% (43) &28.2\% (42) & $-$ &   $-$ \\
 \hline
 \hline
\multirow{2}{*}{21} & \multirow{2}{*}{1} & \multirow{2}{*}{QCC} & Score & 0.39 (45) & 0.82 & 0.86 & 0.80 & 0.61 & 0.73 & 0.58 & 0.92 & 0.94 & 0.94 & 0.00 (0) & $-$ \\
 &  &  &$\Delta$& $-$  & $-$ & 3.9\% (26) & -3.9\% (18) &  $-$  & 18.0\% (42) & -8.3\% (32) & $-$   & 6.6\% (41) & 5.8\% (36) &   $-$&   $-$ \\
 \hline
\multirow{2}{*}{21} & \multirow{2}{*}{1} & \multirow{2}{*}{QCC-I} & Score &  0.28 (7) & 0.57 & 0.59 & 0.54 (49) & 0.48 & 0.54 & 0.48 & 0.94 & 0.96 & 0.98 (49) & 0.00 (0) &  $-$ \\
  &  &  &$\Delta$ &   $-$  &  $-$  & 2.9\% (17) & -5.6\% (15)  &  $-$  &12.1\% (47) & -2.9\% (45) &  $-$  & 5.0\% (31) & 4.9\% (21) &  $-$& $-$  \\
 \hline
\multirow{2}{*}{21} & \multirow{2}{*}{1} & \multirow{2}{*}{QCC-X} & Score  &  0.35 (40) & 0.40 & 0.73 & 0.73 & 0.67 & 0.92 & 0.82 & 0.59 (19) & 0.65 (19) & 0.77 (16) & 0.00 (0) & $-$ \\
  &  &  &$\Delta$& $-$& $-$ & 42.1\% (47) & 42.1\% (49) & $-$ & 27.0\% (45) & 14.5\% (33)  & $-$ & 33.9\% (18) &  49.8\% (15)&  $-$ &  $-$ \\
 \hline
\end{tabular}
\caption{Performance comparison using \textit{plan score} and \textit{\% improvement} ($\Delta$). Plan score ($max=1.00$) uses the formula of the International Planning Competition (IPC): if the best-known makespan for instance $i$ is $P_i$, then for a given solver $X$ that returns a plan $p^i_X$: $\text{Score}(i,X) = P_i / C_{max}(p^i_X)$. For the 50-instance benchmark set, solver score is the average over the instance scores for which a solution was found by the solver. Values in brackets indicate the number of problems that were solved to feasibility (if no brackets, all 50 instances were solved). $\dagger$ indicates the instances were solved to proven optimality. \% improvement ($\Delta$) assesses the average makespan improvement of the hybrid over the best solutions of the stand-alone planner. The number in brackets indicates the number of instances the hybrid approach improved on the stand-alone planner. $|N|{=}8$ problems are run for two minutes, and $|N|{=}21$ for 10 minutes.}
\label{table:results}
\end{table*}}

\section{Empirical Evaluation}
\label{sec:evaluation}

In this section we present an extensive experimental assessment of the stand-alone methods and our hybrids.

\subsection{Setup}
\label{subsec:empirical_setup}

\subsubsection{Problem Instances} We start with the previously studied QCC problem benchmark set \cite{vent:ijcai17} that compiles Quantum Approximate Optimization (QAOA) circuits \cite{Farhi14} for MaxCut to the architecture inspired by the Rigetti Computing's quantum computer~\cite{Sete16} (refer to Figure~\ref{fig:hardware}). As most planners cannot solve $|N| = 40$ qubit problems, we study two problem sizes: $|N| = 8$ and $|N| = 21$. We solve each benchmark instance using one of three problem variations: the baseline problem (QCC), initializations (QCC-I), and crosstalk (QCC-X).

In total, we document results from four temporal planners, CP Optimizer, and our two hybrid approaches on nine sets of 50 problems each, for the total of 5,400 data points.
\subsubsection{Software} In addition to the TFD~\cite{TFD} and LPG~\cite{LPG} temporal planners used in the previous work~\cite{vent:ijcai17}, we include results for two additional planners: CPT~\cite{vidal:aij2006} and POPF~\cite{popf}. CPT is an optimality-focused planner that uses CP inference techniques in a partial-order planning framework. POPF combines forward-state-space search with the partial-order planning framework. We use the commercial software CP Optimizer to represent and solve our CP model.\footnote{Experiments are implemented in C++ on an Intel Core i7-2670QM with 8GB of RAM running Ubuntu 14.04 LTS Linux. We use CP Optimizer version 12.6.3 single-threaded with default search and extended NoOverlap inference (all other inference default). All planners, except CPT, used the same machine. CPT, due to software issues, was run on a RedHat Linux 2.4Ghz machine with 8GB RAM.}

Three of the planners tested (LPG, TFD, and POPF) are anytime planners. These continue to return plans of gradually better quality until the allotted runtime is over. CPT, as designed, returns a single solution at the end of the given run-time, if one is found, which it attempts to prove optimal. CP Optimizer is an anytime solver, improving solution quality over time, and seeking to prove optimality.

\begin{figure*}
\centering
    \includegraphics[width=0.9\textwidth]{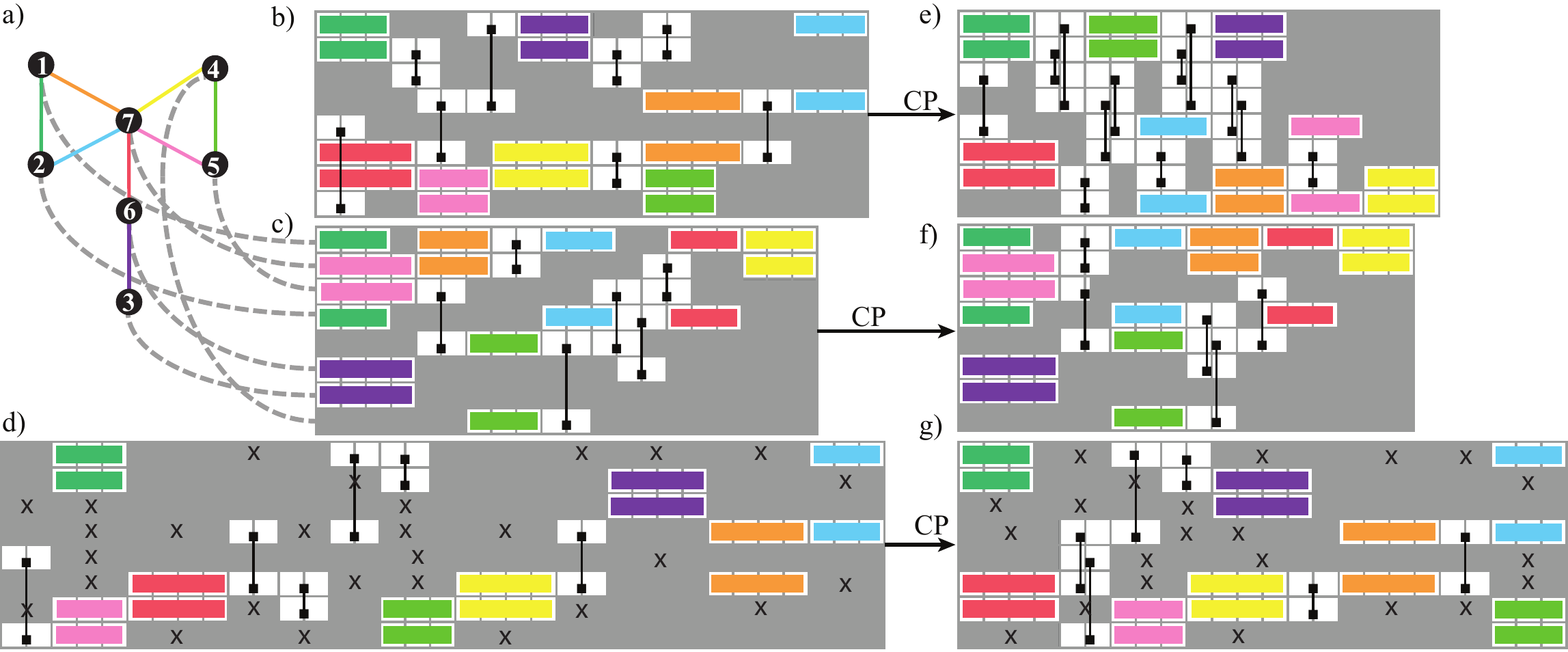}
\caption{$a)$ Problem instance for $|N|=8$ (using 7 qubits) for the QAOA-MaxCut model algorithm. Colored edges represent \ps gate executions between corresponding quantum states. $b)-g)$: Compilation plans obtained for the $\mathcal{P}=1$ problem instance in a). Clock cycles on the horizontal axis and qubit locations arranged on the vertical axis, sorted according to Figure~\ref{fig:hardware}. Colored tasks represent executed \ps gates, and swap gates by white tasks with vertical lines. b) LPG compilation for QCC baseline. c) LPG for QCC-I: gray lines indicate qubit state initialization. d) LPG for QCC-X: crosses identify disabled regions due to crosstalk constraints. e), f) and g), respectively, show the improvements when CP uses b), c), and d) as a warm start.}
\label{fig:schedules}
\end{figure*}

\subsection{Analysis}
\label{subsec:analysis}

Our experimental results are summarized in Table \ref{table:results}. A visualization of solutions to the variants of a QCC problem instance is illustrated in Figure \ref{fig:schedules}.

\subsubsection{Temporal Planning} The anytime temporal planners were able to consistently find solutions for all problems (with the exception of POPF for larger QCC-X problems), while CPT was only able to return solutions on the two problem sets with the smallest anticipated makespans. This is not surprising given that CPT works by bounding the planning horizon, loosening the bound if no plan is found and tightening it otherwise. This multi-step procedure is more time consuming as the optimal plan makespan increases.

TFD was the best overall performer in the previous study \cite{vent:ijcai17} and remains strong  in our experiments.\footnote{For a number of QCC-X problem instances, the best solution returned by TFD was invalid as confirmed by the plan validator VAL. Such solutions were manually removed for this planner until a valid solution was found.} It is the best overall planner for $|N|=8$ problems, for $\mathcal{P}\in \{1,2\}$, though its performance degrades for the larger $|N|=21$ problems. LPG is competitive for single stage problems, $\mathcal{P}{=}1$, however, seems to perform poorly for problems with a mixing stage. Solutions yielded by LPG for the three problem variants on a given problem instance, are illustrated in Figure \ref{fig:schedules} b), c), and d). The POPF planner had relatively consistent performance across all problems, providing notably strong performance for QCC and QCC-I on the larger $|N|=21$ problems.

\subsubsection{Constraint Programming} Overall, the stand-alone CP approach is competitive with temporal planning on the three smaller $|N|{=}8$, $\mathcal{P}{=}1$ problem sets, doing particularly well on QCC and QCC-X problems. The crosstalk constraints within QCC-X increase the scope of the NoOverlap constraints, enhancing inference that the solver can perform and leading to stronger performance on these problems. The QCC-I problem variant proved difficult for CP on all problem sizes, as the lack of initially defined qubit states is more likely to lead the search away from candidate solutions. The large number of optional tasks in larger problems overwhelmed the approach and lead to poor overall performance, indicated by CP's inability to find solutions to all problems.

\subsubsection{Hybrid Approaches}
The \textit{Last} hybrid improves upon every planner on every problem variant. Since this hybrid technique always initiates the CP search from the best planning solution found in the full runtime limit, the final solution found is always equivalent, or better, in quality than the stand-alone planner.

For QCC-X problems, the \textit{Half} hybrid is always beneficial - often significantly so. This result is consistent with those of stand-alone CP, which performed surprisingly well on this problem variant for $|N|{=}8$. When temporal planning produces poor solutions (0.4-0.7 plan score), the halfway switching policy is often (but not always) able to find plans of up to 49.8\% better quality. On problems that temporal planning produces high quality solutions (0.85+ plan score), the halfway switching policy is always beneficial, indicating our CP warm-start is useful for ``fine-tuning'' high quality plans. Though displaying strong performance on the majority of the problem classes, the halfway switching policy has weaker performance (\textit{w.r.t.} \% improvement) when hybridized with the LPG and TFD planners for QCC and QCC-I problems with $|N|{=}21$. We note, however, that it actually improves the solution value for the majority of these instances in the case of TFD (32 and 45 instances for QCC and QCC-I, respectively). The \textit{Half} hybrid often produces similar scores to our oracle-based hybrid and sometimes outperforms it, indicating the development of a more sophisticated switching policy may prove worthwhile.

Figure \ref{fig:schedules} illustrates CP improvement on seemingly reasonable quality plans; e), f), and g) all exhibiting improvement over their planning solution warm starts. The plan detailed in e) is notable in two ways: i) it shows that more swap gate tasks can result in a reduced makespan, and ii) it shows that the CP portion of the hybrid does more than a simple rescheduling of tasks, often changing the qubit pairs that \ps gates are applied to, and thus using different actions than those supplied by the temporal planning warm start.

\section{Conclusions \& Future Work}
\label{sec:conclusion}

In this paper, we investigate CP as an alternative to temporal planning for QCC. Our empirical work shows that stand-alone CP does not scale as well as current state-of-the-art temporal planners as problems increase in size, however, a hybrid, where CP is warm started with a solution found by temporal planning, out-performs both planning and CP alone for the majority of problems. We introduce new variations of QCC that include the ability to arbitrarily initialize qubits (QCC-I) and to account for crosstalk interactions between qubits (QCC-X). These variants produce a more diverse QCC benchmark for which the planning technologies exhibited substantially different performance variance.

Given these results, our work strengthens the message that AI planning is a suitable technology to address QCC challenges and demonstrates the benefit of integrating it with alternate optimization methods for building effective compilers of real-world quantum hardware. Our proposed hybridization is an initial investigation of how to best combine temporal planning and CP for QCC problems. In a future work, following \citeauthor{beck2011combining} (\citeyear{beck2011combining}), we will analyze the performance of the two components of our hybrid in order to determine an adaptive schedule for switching between temporal planning and CP.

\section{Acknowledgments}
This research has been supported by NASA Academic Mission Services, contract number NNA16BD14C, NASA Advanced Exploration Systems program, and the Natural Sciences and Engineering Research Council of Canada. We thank Amanda Coles, Alfonso Gerevini, Chiara Piacentini, and Vincent Vidal for their help with the temporal planners.

\newpage

\bibliographystyle{aaai}
\bibliography{aaai}

\end{document}